\documentclass[copyright,creativecommons]{eptcs}
\providecommand{\event}{GandALF 2013}
\usepackage[utf8]{inputenc}
\usepackage[T1]{fontenc}
\usepackage{breakurl}             
\usepackage{amsmath}
\usepackage{amsthm}
\usepackage{amssymb}
\usepackage[ruled, noline]{algorithm2e}
\usepackage{microtype}

\newtheorem{theorem}{Theorem}
\newtheorem{lemma}[theorem]{Lemma}

\newtheorem{proposition}[theorem]{Proposition}

\newcommand{\mult}{\otimes}
\newcommand{\dist}{d}
\newcommand{\val}{\mu}
\newcommand{\appr}{\hat{\mu}}
\newcommand{\vt}{\delta}
\newcommand{\at}{\hat{\delta}}
\newcommand{\wm}{D} 
\newcommand{\awm}{F} 
\newcommand{\algidx}{s}

\newcommand{\polylog}{\operatorname{polylog}}

\DeclareMathOperator*{\argmin}{arg\,min}
\DeclareMathOperator*{\amp}{approx-min-plus}

\title{Approximating the minimum cycle mean}

\author{Krishnendu Chatterjee\thanks{Supported by the Austrian Science Fund (FWF): P23499-N23 and S11407-N23 (RiSE),
an ERC Start Grant (279307: Graph Games), and a Microsoft Faculty Fellows Award.}
\institute{IST Austria \\Institute of Science and Technology \\ Klosterneuburg, Austria}
\and 
Monika Henzinger\thanks{Supported by the Austrian Science Fund (FWF): P23499-N23, the Vienna Science and Technology Fund (WWTF) grant ICT10-002, and the University of Vienna (IK \mbox{I049-N}).}, 
Sebastian Krinninger\thanks{Supported by the Austrian Science Fund (FWF): P23499-N23 and the University of Vienna (IK \mbox{I049-N}).}, 
Veronika Loitzenbauer\thanks{Supported by the Vienna Science and Technology Fund (WWTF) grant ICT10-002.}
\institute{University of Vienna \\ Faculty of Computer Science \\ Vienna, Austria}
}

\def\titlerunning{Approximating the minimum cycle mean}
\def\authorrunning{K. Chatterjee, M. Henzinger,\; \\ S. Krinninger \& V. Loitzenbauer}
\begin{document}
\maketitle

\begin{abstract}
We consider directed graphs where each edge is labeled with an 
integer weight and study the fundamental algorithmic question of computing 
the value of a cycle with minimum mean weight.
Our contributions are twofold:
(1)~First we show that the algorithmic question is reducible in $O(n^2)$ time
to the problem of a logarithmic number of \emph{min-plus} matrix multiplications of 
$n\times n$-matrices, where $n$ is the number of vertices of the graph.
(2)~Second, when the weights are nonnegative, we present the first 
$(1 + \epsilon)$-approximation algorithm for the problem and the running time of our 
algorithm is $\widetilde{O} (n^\omega \log^3{(nW/\epsilon)} / \epsilon)$\footnote{The $\widetilde{O}$-notation hides a polylogarithmic factor.}, 
where $O (n^\omega)$ is the time required for the \emph{classic} 
$n\times n$-matrix multiplication and $W$ is the maximum value
of the weights.
\end{abstract}
\section{Introduction}

\noindent{\bf Minimum cycle mean problem.} 
We consider a fundamental graph algorithmic problem of computing the value 
of a minimum mean-weight cycle in a finite directed graph.
The input to the problem is a directed graph $G=(V,E,w)$ with a finite set 
$V$ of $n$ vertices, $E$ of $m$ edges, and a weight function $w$ that assigns 
an integer weight to every edge. 
Given a cycle $C$, the mean weight $\val(C)$ of the cycle is the ratio 
of the sum of the weights of the cycle and the number of edges in the cycle.
The algorithmic question asks to compute 
$\val=\min\{\val(C) \mid C \text{ is a cycle} \}$: the minimum cycle mean.
The minimum cycle mean problem is an important problem in 
combinatorial optimization and has a long history of algorithmic study.
An $O(n m)$-time algorithm for the problem was given by Karp~\cite{Karp1978};
and the current best known algorithm for the problem, which is over 
two decades old, by Orlin and Ahuja require $O(m \sqrt{n} \log{(n W)})$ 
time~\cite{OrlinA1992}, where $W$ is the maximum absolute value of the 
weights.

\smallskip\noindent{\bf Applications.} 
The minimum cycle mean problem is a basic combinatorial optimization 
problem that has numerous applications in network flows~\cite{AhujaMO1992}.
In the context of formal analysis of reactive systems, the performance of  
systems as well as the average resource consumption of systems is modeled
as the minimum cycle mean problem.
A reactive system is modeled as a directed graph, where vertices represent 
states of the system, edges represent transitions, and every edge is assigned a 
\emph{nonnegative} integer representing the resource consumption 
(or delay) associated with the transition.
The computation of a minimum average resource consumption behavior 
(or minimum average response time) corresponds to the computation of
the minimum cycle mean.
Several recent works model other quantitative aspects of system
analysis (such as robustness) also as the mean-weight problem (also known as \emph{mean-payoff objectives})~\cite{BGHJ09,DM10}.

\smallskip\noindent{\bf Results.} This work contains the following results.
\begin{enumerate}
\item \emph{Reduction to min-plus matrix multiplication.}
We show that the minimum cycle mean problem is reducible in $O(n^2)$ 
time to the problem of a logarithmic number of min-plus matrix multiplications of 
$n\times n$-matrices, where $n$ is the number of vertices of the graph.
Our result implies that algorithmic improvements for min-plus matrix
multiplication will carry over to the minimum cycle mean problem
with a logarithmic multiplicative factor and $O(n^2)$ additive
factor in the running time.

\item \emph{Faster approximation algorithm.}
When the weights are nonnegative, we present the first 
$(1 + \epsilon)$-approx\-imation algorithm for the problem that outputs 
$\appr$ such that 
$\val \leq \appr \leq (1+\epsilon) \val$ and the running time of 
our algorithm is $\widetilde{O} (n^\omega \log^3{(nW/\epsilon)} / \epsilon)$. 
As usual, the $\widetilde{O}$-notation is used to ``hide'' a polylogarithmic factor, i.e., $\widetilde{O} (T (n, m, W)) = O (T (n, m, W) \cdot \polylog(n))$, and $O (n^\omega)$ is the time required for the \emph{classic} 
$n\times n$-matrix multiplication.
The current best known bound for $\omega$ is $\omega < 2.3727$.
The worst case complexity of the current best known algorithm
for the minimum cycle mean problem is $O(m \sqrt{n} \log{(n W)})$~\cite{OrlinA1992}, 
which could be as bad as $O(n^{2.5} \log{(n W)})$. Thus for $(1 + \epsilon)$-approximation 
our algorithm provides better dependence in $n$. 
Note that in applications related to systems analysis the weights are always
nonnegative (they represent resource consumption, delays, etc);
and the weights are typically small, whereas the state space of 
the system is large.
Moreover, due to imprecision in modeling, approximations in weights
are already introduced during the modeling phase.
Hence $(1 + \epsilon)$-approximation of the minimum cycle mean problem 
with small weights and large graphs is a very relevant algorithmic problem 
for reactive system analysis, and we improve the long-standing 
complexity of the problem.

The key technique that we use to obtain the approximation algorithm is a
combination of the value iteration algorithm for the minimum cycle mean 
problem, and a technique used for an approximation algorithm for 
all-pair shortest path problem for directed graphs.
Table~\ref{tab:running_time_comparison} compares our algorithm with the 
asymptotically fastest existing algorithms.
\end{enumerate}

\begin{table}[htbp]
\centering
\begin{tabular}{c | c | c | c}
\textbf{Reference} & \textbf{Running time} & \textbf{Approximation} & \textbf{Range} \\
\hline
Karp~\cite{Karp1978} & $ O (m n) $ & exact & $ [-W, W] $ \\
Orlin and Ahuja~\cite{OrlinA1992} & $ O (m \sqrt{n} \log{(n W)}) $ & exact & $ [-W, W] \cap \mathbb{Z} $ \\
Sankowski~\cite{Sankowski2005} (implicit) & $ \widetilde{O} (W n^\omega \log{(n W)}) $ & exact & $ [-W, W] \cap \mathbb{Z} $ \\
Butkovic and Cuninghame-Green~\cite{ButkovicG1992} & $ O (n^2) $ & exact & $ \{ 0, 1 \} $ \\
This paper & $ \widetilde{O} (n^\omega \log^3{(nW/\epsilon)} / \epsilon) $ & $ 1+\epsilon $ & $ [0, W] \cap \mathbb{Z} $ \\
\end{tabular}
\caption{Current fastest asymptotic running times for computing the minimum cycle mean}
\label{tab:running_time_comparison}
\end{table}

\subsection{Related work}

The minimum cycle mean problem is basically equivalent to solving a deterministic Markov decision process (MDP)~\cite{ZwickP1996}.
The latter can also be seen as a single-player mean-payoff game~\cite{EhrenfeuchtM1979, GurvichKK1988, ZwickP1996}.
We distinguish two types of algorithms: algorithms that are independent of the weights of the graph and algorithms that depend on the weights in some way.
By $ W $ we denote the maximum absolute edge weight of the graph.

\smallskip\noindent{\bf Algorithms independent of weights.} 
The classic algorithm of Karp~\cite{Karp1978} uses a dynamic programming approach to find the minimum cycle mean and runs in time $ O (m n) $.
The main drawback of Karp's algorithm is that its best-case and worst-case running times are the same.
The algorithms of Hartmann and Orlin~\cite{HartmannO1993} and of Dasdan and Gupta~\cite{DasdanG1998} address this issue, but also have a worst-case complexity of $ O(m n) $.
By solving the more general parametric shortest path problem, Karp and Orlin~\cite{KarpO1981} can compute the minimum cycle mean in time $ O (m n \log{n}) $.
Young, Tarjan, and Orlin~\cite{YoungTO1991} improve this running time to $ O (m n + n^2 \log{n}) $.

A well known algorithm for solving MDPs is the value iteration algorithm.
In each iteration this algorithm spends time $ O (m) $ and in total it performs $ O(nW) $ iterations.
Madani~\cite{Madani2002} showed that, for \emph{deterministic} MDPs (i.e., weighted graphs for which we want to find the minimum cycle mean), a certain variant of the value iteration algorithm ``converges'' to the optimal cycle after $ O (n^2) $ iterations which gives a running time of $ O(m n^2) $ for computing the minimum cycle mean.
Using similar ideas he also obtains a running time of $ O(m n) $.
Howard's policy iteration algorithm is another well-known algorithm for solving MDPs~\cite{Howard1960}.
The complexity of this algorithm for deterministic MDPs is unresolved.
Recently, Hansen and Zwick~\cite{HansenZ2010} provided a class of weighted graphs on which Howard's algorithm performs $ \Omega (n^2) $ iterations where each iteration takes time $ O (m) $.

\smallskip\noindent\textbf{Algorithms depending on weights.}
If a graph is complete and has only two different edge weights, then the minimum cycle mean problem problem can be solved in time $ O(n^2) $ because the matrix of its weights is bivalent~\cite{ButkovicG1992}.

Another approach is to use the connection to the problem of detecting a negative cycle.
Lawler~\cite{Lawler1976} gave a reduction for finding the minimum cycle mean that performs $ O(\log (n W)) $ calls to a negative cycle detection algorithm.
The main idea is to perform binary search on the minimum cycle mean.
In each search step the negative cycle detection algorithm is run on a graph with modified edge weights.
Orlin and Ahuja~\cite{OrlinA1992} extend this idea by the approximate binary search technique~\cite{Zemel1987}.
By combining approximate binary search with their scaling algorithm for the assignment problem they can compute the minimum cycle mean in time $ O(m \sqrt{n} \log{n W}) $.

Note that in its full generality the single-source shortest paths problem (SSSP) also demands the detection of a negative cycle reachable from the source vertex.\footnote{Remember that, for example, Dijkstra's algorithm for computing single-source shortest paths requires non-negative edge weights which excludes the possibility of negative cycles.}
Therefore it is also possible to reduce the minimum cycle mean problem to SSSP.
The best time bounds on SSSP are as follows.
Goldberg's scaling algorithm~\cite{Goldberg1995} solves the SSSP problem (and therefore also the negative cycle detection problem) in time $ O (m \sqrt{n} \log{W}) $.
McCormick~\cite{McCormick1993} combines approximate binary search with Goldberg's scaling algorithm to find the minimum cycle mean in time $ O (m \sqrt{n} \log{n W}) $, which matches the result of Orlin and Ahuja~\cite{OrlinA1992}.
Sankowski's matrix multiplication based algorithm~\cite{Sankowski2005} solves the SSSP problem in time $ \widetilde O (W n^\omega) $.
By combining binary search with Sankowski's algorithm, the minimum cycle mean problem can be solved in time $ \widetilde O (W n^\omega \log{nW}) $

\smallskip\noindent\textbf{Approximation of minimum cycle mean.}
To the best of our knowledge, our algorithm is the first approximation algorithm specifically for the minimum cycle mean problem.
There are both additive and multiplicative fully polynomial-time approximation schemes for solving mean-payoff games~\cite{RothBKM2010, BorosEFGMM2011}, which is a more general problem.
Note that in contrast to finding the minimum cycle mean it is not known whether the exact solution to a mean-payoff game can be computed in polynomial time.
The results of~\cite{RothBKM2010} and~\cite{BorosEFGMM2011} are obtained by reductions to a pseudo-polynomial algorithm for solving mean-payoff games.
In the case of the minimum cycle mean problem, these reductions do not provide an improvement over the current fastest exact algorithms mentioned above.

\smallskip\noindent\textbf{Min-plus matrix multiplication.}
Our approach reduces the problem of finding the minimum cycle mean to computing the (approximate) min-plus product of matrices.
The naive algorithm for computing the min-plus product of two matrices runs in time $ O(n^3) $.
To date, no algorithm is known that runs in time $ O(n^{3-\alpha}) $ for some $ \alpha > 0 $, so-called \emph{truly subcubic} time.
This is in contrast to classic matrix multiplication that can be done in time $ O (n^\omega) $ where the current best bound on $ \omega $ is $ \omega < 2.3727 $~\cite{Williams2012}.
Moreover, Williams and Williams~\cite{WilliamsW2010} showed that computing the min-plus product is computationally equivalent to a series of problems including all-pairs shortest paths and negative triangle detection.
This provides evidence for the hardness of these problems.
Still, the running time of $ O(n^3) $ for the min-plus product can be improved by logarithmic factors and by assuming small integer entries.

Fredman~\cite{Fredman1976} gave an algorithm for computing the min-plus product with a slightly subcubic running time of $ O(n^3 (\log\log{n})^{1/3} / (\log{n})^{1/3}) $.
This algorithm is ``purely combinatorial'', i.e., it does not rely on fast algorithms for classic matrix multiplication.
After a long line of improvements, the current fastest such algorithm by Chan~\cite{Chan2010} runs in time $ O(n^3 (\log{\log{n}})^3 / (\log{n})^2) $.

A different approach for computing the min-plus product of two integer matrix is to reduce the problem to classic matrix multiplication~\cite{Yuval1976}.
In this way, the min-plus product can be computed in time $ O (M n^\omega \log{M}) $ which is pseudo-polynomial since $ M $ is the maximum absolute integer entry~\cite{AlonGM1997}.
This observation was used by Alon, Galil, and Margalit~\cite{AlonGM1997} and Zwick~\cite{Zwick2002} to obtain faster all-pairs shortest paths algorithms in directed graphs for the case of small integer edge weights.
Zwick also combines this min-plus matrix multiplication algorithm with an adaptive scaling technique that allows to compute $ (1+\epsilon) $-approximate all-pairs shortest paths in graphs with non-negative edge weights.
Our approach of finding the minimum cycle mean extensively uses this technique.

\section{Definitions}

Throughout this paper we let $ G = (V, E, w) $ be a weighted directed graph with a finite set of vertices $ V $ and a set of edges $ E $ such that every vertex has at least one outgoing edge.
The weight function $ w $ assigns a nonnegative integer weight to every edge.
We denote by $ n $ the number of vertices of $ G $ and by $ m $ the number of edges of $ G $.
Note that $ m \geq n $ because every vertex has at least one outgoing edge.

A \emph{path} is a finite sequence of edges $ P = (e_1, \ldots, e_k) $ such that for all consecutive edges $ e_i = (x_i, y_i) $ and $ e_{i+1} = (x_{i+1}, y_{i+1}) $ of $ P $ we have $ y_i = x_{i+1} $.
Note that edges may be repeated on a path, we \emph{do not} only consider simple paths.
A \emph{cycle} is a path in which the start vertex and the end vertex are the same.
The \emph{length of a path $ P $} is the number of edges of $ P $.
The \emph{weight of a path $ P = (e_1, \ldots, e_k) $}, denoted by $ w (P) $ is the sum of its edge weights, i.e. $ w(P) = \sum_{1 \leq i \leq k} w (e_i) $.

The \emph{minimum cycle mean} of $ G $ is the minimum mean weight of any cycle in $ G $.
For every vertex $ x $ we denote by $ \val (x) $ the value of the minimum mean-weight cycle reachable from $ x $.
The minimum cycle mean of $ G $ is simply the minimum $ \val (x) $ over all vertices $ x $.
For every vertex $ x $ and every integer $ t \geq 1 $ we denote by $ \vt_t (x) $ the minimum weight of all paths starting at $ x $ that have length $ t $, i.e., consist of exactly $ t $ edges.
For all pairs of vertices $ x $ and $ y $ and every integer $ t \geq 1 $ we denote by $ \dist_t (x, y) $ the minimum weight of all paths of length $ t $ from $ x $ to $ y $.
If no such path exists we set $ \dist_t (x, y) = \infty $.

For every matrix $ A $ we denote by $ A[i, j] $ the entry at the $i$-th row and the $j$-th column of $ A $.
We only consider $ n \times n $ matrices with integer entries, where $ n $ is the size of the graph.
We assume that the vertices of $ G $ are numbered consecutively from $ 1 $ to $ n $, which allows us to use $ A[x, y] $ to refer to the entry of $ A $ belonging to vertices $ x $ and $ y $.
The \emph{weight matrix $ \wm $ of $ G $} is the matrix containing the weights of $ G $.
For all pairs of vertices $ x $ and $ y $ we set $ \wm [x, y] = w (x, y) $ if the graph contains the edge $ (x, y) $ and $ \wm [x, y] = \infty $ otherwise.

We denote the \emph{min-plus product} of two matrices $ A $ and $ B $ by $ A \mult B $.
The min-plus product is defined as follows.
If $ C = A \mult B $, then for all indices $ 1 \leq i, j \leq n $ we have $ C[i,j] = \min_{1 \leq k \leq n} (A[i, k] + B[k, j]) $.
We denote by $ A^t $ the $ t $-th power of the matrix $ A $.
Formally, we set $ A^1 = A $ and $ A^{t+1} = A \mult A^t $ for $ t \geq 1 $.
We denote by $ \omega $ the exponent of classic matrix multiplication, i.e., the product of two $ n \times n $ matrices can be computed in time $ O (n^{\omega}) $.
The current best bound on $ \omega $ is $ \omega < 2.3727 $~\cite{Williams2012}.

\section{Reduction of minimum cycle mean to min-plus matrix multiplication}\label{sec:reduction}

In the following we explain the main idea of our approach which is to use min-plus matrix multiplication to find the minimum cycle mean.
The well-known value iteration algorithm uses a dynamic programming approach to compute in each iteration a value for every vertex $ x $ from the values of the previous iteration.
After $ t $ iterations, the value computed by the value iteration algorithm for vertex $ x $ is equal to $ \vt_t (x) $, the minimum weight of all paths with length $ t $ starting at $ x $.
We are actually interested in $ \val (x) $, the value of the minimum mean-weight cycle reachable from $ x $.
It is well known that $ \lim_{t \to \infty} \vt_t (x) / t = \val (x) $ and that the value of $ \val (x) $ can be computed from $ \vt_t (x) $ if $ t $ is large enough~$(t = O (n^3 W) )$~\cite{ZwickP1996}.\footnote{Specifically, for $ t = 4 n^3 W$ the unique number in $\left(\vt_t (x) / t - 1/[2n(n-1)], \vt_t (x) / t + 1 / [2n(n-1)]\right) \cap \mathbb{Q}$ that has a denominator of at most $n$ is equal to $\val(x)$~\cite{ZwickP1996}.}
Thus, one possibility to determine $ \val (x) $ is the following: first, compute $ \vt_t (x) $ for $ t $ large enough with the value iteration algorithm and then compute $ \val (x) $ from $ \vt_t (x) $.
However, using the value iteration algorithm for computing $ \vt_t (x) $ is expensive because its running time is linear in $ t $ and thus pseudo-polynomial.

Our idea is to compute $ \vt_t (x) $ for a large value of $ t $ by using fast matrix multiplication instead of the value iteration algorithm.
We will compute the matrix $ \wm^t $, the $t$-th power of the weight matrix (using min-plus matrix multiplication).
The matrix $ \wm^t $ contains the value of the minimum-weight path of length exactly $ t $ for all pairs of vertices.
Given $ \wm^t $, we can determine the value $ \vt_t (x) $ for every vertex $ x $ by finding the minimum entry in the row of $ \wm^t $ corresponding to $ x $.
\begin{proposition}\label{prop:rowmin}
For every $ t \geq 1 $ and all vertices $ x $ and $ y $ we have (i) $ \dist_t (x, y) = \wm^t [x, y] $ and (ii) $ \vt_t (x) = \min_{y \in V} \wm^t [x, y] $.
\end{proposition}
\begin{proof}
We give the proof for the sake of completeness.
The claim $ \dist_t (x, y) = \wm^t [x, y] $ follows from a simple induction on $ t $.
If $ t = 1 $, then clearly the minimal-weight path of length $ 1 $ from $ x $ to $ y $ is the edge from $ x $ to $ y $ if it exists, otherwise $ \dist_t (x, y) = \infty $.
If $ t \geq 1 $, then a minimal-weight path of length $ t $ from $ x $ to $ y $ (if it exists) consists of some outgoing edge of $ e = (x, z) $ as its first edge and then a minimal-weight path of length $ t - 1 $ from $ z $ to $ y $.
We therefore have $ \dist_t (x, y) = \min_{(x, z) \in E} w (x, z) + \dist_{t-1} (z, y) $.
By the definition of the weight matrix and the induction hypothesis we get $ \dist_t (x, y) = \min_{z \in V} \wm [x, z] + \wm^{t-1} [z, y] $.
Therefore the matrix $ \wm \mult \wm^{t-1} = \wm^t $ contains the value of $ \dist_t (x, y) $ for every pair of vertices $ x $ and $ y $.

For the second claim, $ \vt_t (x) = \min_{y \in V} \wm^t [x, y] $, observe that by the definition of $ \vt_t (x) $ we obviously have $ \vt_t (x) = \min_{y \in V} \dist_t (x, y) $ because the minimal-weight path of length $ t $ starting at $ x $ has \emph{some} node $ y $ as its end point.
\end{proof}

Using this approach, the main question is how fast the matrix $ \wm^t $ can be computed.
The most important observation is that $ \wm^t $ (and therefore also $ \vt_t (x) $) can be computed by repeated squaring with only $ O (\log t) $ min-plus matrix multiplications.
This is different from the value iteration algorithm, where $ t $ iterations are necessary to compute $ \vt_t (x) $.
\begin{proposition}\label{prop:logt}
For every $ t \geq 1 $ we have $ \wm^{2t} = \wm^t \mult \wm^t $.
Therefore the matrix $ \wm^t $ can be computed with $ O (\log t) $ many min-plus matrix multiplications.
\end{proposition}
\begin{proof}
We give the proof for the sake of completeness.
It can easily be verified that the min-plus matrix product is associative~\cite{AhoHU1974} and therefore $ \wm^{2t} = \wm^t \mult \wm^t $.
Therefore, if $ t $ is a power of two, we can compute $ \wm^{t} $ with $ \log t $ min-plus matrix multiplications.
If $ t $ is not a power of two, we can decompose $ \wm^t $ into $ \wm^t = \wm^{t_1} \mult \ldots \mult \wm^{t_k} $ where each $ t_i \leq t $ (for $ 1 \leq i \leq k $) is a power of two and $ k \leq \lceil \log t \rceil $.
By storing intermediate results, we can compute $ \wm^{2^i} $ for every $ 0 \leq i \leq \lceil \log t \rceil $ with $ \lceil \log t \rceil $ min-plus matrix multiplications.
Using the decomposition above, we have to multiply at most $ \lceil \log t \rceil $ such matrices to obtain $ \wm^t $.
Therefore the total number of min-plus matrix multiplications needed for computing $ \wm^t $ is $ O (\log t) $.
\end{proof}

The running time of this algorithm depends on the time needed for computing the min-plus product of two integer matrices.
This running time will usually depend on the two parameters $ n $ and $ M $ where $ n $ is the size of the $ n \times n $ matrices to be multiplied (in our case this is equal to the number of vertices of the graph) and the parameter $ M $ denotes the maximum absolute integer entry in the matrices to be multiplied.
When we multiply the matrix $ \wm $ by itself to obtain $ \wm^2 $, we have $ M = W $, where $ W $ is the maximum absolute edge weight.
However, $ M $ increases with every multiplication and in general, we can bound the maximum absolute integer entry of the matrix $ \wm^t $ only by $ M = t W $.
Note that $O(n^2)$ operations are necessary to extract the minimum cycle mean $\val(x)$ for all vertices $x$ from the matrix $\wm^t$.

\begin{theorem}
If the min-plus product of two $ n \times n $ matrices with entries in $ \{ -M, \ldots, -1, 0, 1, \ldots, M, \infty \} $ can be computed in time $ T (n, M) $, then the minimum cycle mean problem can be solved in time $ T (n, t W) \log t$ where $ t = O (n^3 W) $.\footnote{Note that necessarily $ T (n, M) = \Omega (n^2) $ because the result matrix has $ n^2 $ entries that have to be written.}
\end{theorem}

Unfortunately, the approach outlined above does not immediately improve the running time for the minimum cycle mean problem because min-plus matrix multiplication currently cannot be done fast enough.
However, our approach is still useful for solving the minimum cycle mean problem \emph{approximately} because approximate min-plus matrix multiplication can be done faster than its exact counterpart.

\section{Approximation algorithm}
In this section we design an algorithm that computes an approximation of the minimum cycle mean in graphs with nonnegative integer edge weights.
It follows the approach of reducing the minimum cycle mean problem to min-plus matrix multiplication outlined in Section~\ref{sec:reduction}.
The key to our algorithm is a fast procedure for computing the min-plus product of two integer matrices approximately.
We will proceed as follows.
First, we explain how to compute an approximation $ \awm $ of $ \wm^t $, the $t$-th power of the weight matrix $ \wm $.
From this we easily get, for every vertex $ x $, an approximation $ \at_t (x) $ of $ \vt_t (x) $, the minimum-weight of all paths of length $ t $ starting at $ x $.
We then argue that for $ t $ large enough (in particular $ t = O (n^2 W / \epsilon) $), the value $ \vt_t (x) / t $ is an approximation of $ \val (x) $, the minimum cycle mean of cycles reachable from $ x $.
By combining both approximations we can show that $ \at_t (x) / t $ is an approximation of $ \val (x) $.
Thus, the main idea of our algorithm is to compute an approximation of $ \wm^t $ for a large enough $ t $.

\subsection{Computing an approximation of $ \wm^t $}

Our first goal is to compute an approximation of the matrix $ \wm^t $, the $t$-th power of the weight matrix $ \wm $, given $ t \geq 1 $.
Zwick provides the following algorithm for approximate min-plus matrix 
multiplication.
\begin{theorem}[Zwick \cite{Zwick2002}]\label{th:amp}
	Let $A$ and $B$ be two $n \times n$ matrices with integer entries in $[0,M]$
	and let $C := A \mult B$. Let $R \ge \log n$ be a power of two. The algorithm 
	$\amp(A, B, M, R)$ computes the approximate min-plus product $\overline{C}$ of $A$ 
	and $B$ in time\footnote{The running time of $\amp$ is given by $O(n^{\omega}\log M)$
	times the time needed to multiply two $O(R \log n)$-bit integers. With the 
	Schönhage-Strassen algorithm for large integer multiplication, two $k$-bit integers can be multiplied in 
	$O(k \log k \log \log k)$ time, which gives a running time of
	$O(n^{\omega} R \log (M) \log (n) \log(R \log n) \log\log(R \log n))$. 
	This can be bounded by the running time given in Theorem~\ref{th:amp} 
	if $ R \geq \log n $, which will always be the case in the following.}
	$O(n^{\omega} R \log (M) \log^2 (R) \log (n))$ such that for every 
	$1 \le i, j \le n$ it holds that $C[i,j] \le \overline{C}[i,j] \le (1 + 4/R) C[i,j]$.
\end{theorem}

We now give a modification (see Algorithm~\ref{alg}) of Zwick's algorithm for 
approximate shortest paths~\cite{Zwick2002} such that the algorithm computes 
a $(1 + \epsilon)$-approximation $\awm$ of $\wm^t$ 
when $t$ is a power of two such that for $1 \le i, j \le n$ we have $\wm^t[i,j] \le \awm[i,j]
\le (1 + \epsilon) \wm^t[i,j]$.
Just as we can compute $\wm^t$ exactly with $\log t$ min-plus matrix
multiplications, the algorithm computes the $(1 + \epsilon)$-approximation 
of $\wm^t$ in $\log t$ iterations. However, in each iteration only an 
approximate min-plus product is computed. Let $\awm_{\algidx}$ be the approximation of 
$\wm_{\algidx} := \wm^{2^{\algidx}}$. In the ${\algidx}$-th iteration we use 
$\amp(\awm_{{\algidx}-1}, \awm_{{\algidx}-1}, t W, R)$ 
to calculate $\awm_{\algidx}$ with $R$ chosen beforehand such that the desired error bound 
is reached for $\awm = \awm_{ \log t }$. 

\begin{algorithm}
	\caption{Approximation of $\wm^t$}
	\label{alg}
	\SetKwInOut{Input}{input}\SetKwInOut{Output}{output}
	\DontPrintSemicolon
	\BlankLine
	\Input{weight matrix $\wm$, error bound $\epsilon$, $t$ (a power of 2)}
	\Output{$(1 + \epsilon)$-approximation of $\wm^t$}
	\BlankLine
	$\awm \leftarrow \wm$\;
	$r \leftarrow 4 \log t  / \ln(1 + \epsilon)$\;
	$R \leftarrow 2^{\lceil \log r \rceil}$\;
	\For{$\log t$ times}{
		$\awm \leftarrow \amp(\awm, \awm, 2 t W, R)$\;
	}
	\KwRet{$\awm$}
\end{algorithm}

\begin{lemma}\label{lem:alg}
	Given an $0 < \epsilon \le 1$ and a power of two $ t \geq 1 $, Algorithm~\ref{alg} computes a 
	$(1+\epsilon)$-approximation $\awm$ of $\wm^t$ 
	in time
	\begin{equation*}
		O\left(n^{\omega} \cdot 
		\frac{\log^2 (t) }{\epsilon} \cdot 
		\log \left( t W \right) 
		\log^2 \left( \frac{\log (t)}{\epsilon} \right)
		\log (n) 
		\right)
		= \widetilde O \left( n^{\omega} \cdot 
		\frac{\log^2 (t) }{\epsilon} \cdot 
		\log \left( t W \right)  \right)\,
	\end{equation*}
	such that $\wm^t[i,j] \le \awm[i,j] \le \left(1 + \epsilon \right) \wm^t[i,j]$ for all $1 \le i,j \le n.$
\end{lemma}

\begin{proof}
The basic idea is as follows.
The running time of $\amp$ depends linearly on $R$ and logarithmically on $M$,
the maximum entry of the input matrices.
Algorithm~\ref{alg} calls $\amp$ $ \log t $ times. 
Each call increases the error by a factor of $(1 + 4/R)$. However, as only $\log t $
approximate matrix multiplications are used, setting $R$ to the smallest power of
2 that is larger than $4 \log (t) / \ln(1 + \epsilon)$
suffices to bound the approximation error by $(1 + \epsilon)$.
We will show that $2 t W$ is an upper bound on the entries in the input matrices for $\amp$.
The stated running time follows directly from these two facts and Theorem~\ref{th:amp}.

Let $\awm_{\algidx}$ be the approximation of $\wm_{\algidx} := \wm^{2^{\algidx}}$ computed by the algorithm after iteration $ \algidx $.
Recall that $2^{\algidx} W$ is an upper bound on the maximum entry in 
$\wm_{\algidx}$. As we will show, all entries in $\awm_{\algidx}$ are at most $(1 + \epsilon)$-times
the entries in $\wm_{\algidx}$. Since we assume $ \epsilon \leq 1 $, we have $ 1 + \epsilon \leq 2 $. 
Thus $2^{\algidx+1} W$ is an upper bound on the entries in $\awm_{\algidx}$.
Hence $2 t W$ is an upper bound on the entries of $\awm_{\algidx}$ with 
$1 \le {\algidx} <  \log t$, i.e., for all input matrices of $\amp$ in our algorithm.

This results in an overall running time of
\begin{align*}
& O\left(n^{\omega} R \log \left( t W \right) \log (R) \log\log (R) \log (n) \cdot \log (t) \right)\\
& = O\left(n^{\omega} \cdot \frac{\log^2 (t) }{\log(1+\epsilon)} \cdot \log \left( t W \right) \log^2 \left( \frac{\log (t)}{\log(1+\epsilon)} \right) \log (n)\right)\\ 
& = O\left(n^{\omega} \cdot \frac{\log^2 (t) }{\epsilon} \cdot \log \left( t W \right) \log^2 \left( \frac{\log (t)}{\epsilon} \right) \log (n)\right)\,.
\end{align*}
The last equation follows from the inequality $ 1 / \ln(1+\epsilon) \leq (1 + \epsilon) / \epsilon $ for $\epsilon > 0$.
Since $ \epsilon \leq 1 $ it follows that $ 1 / \log(1+\epsilon) = O (1 / \epsilon) $.

	To show the claimed approximation guarantee, we will prove that the inequality
	\begin{equation*}
		\wm_{\algidx}[i,j] \le \awm_{\algidx}[i,j] \le \left(1 + \frac{4}{R} \right)^{\algidx} \wm_{\algidx}[i,j]\,.
	\end{equation*}
	holds after the ${\algidx}$-th iteration of Algorithm~\ref{alg}
	by induction on ${\algidx}$. Note that the $(1+\epsilon)$-approximation 
	follows from this inequality because the parameter $R$ is chosen such that 
	after the $(\log t)$-th iteration of the algorithm it holds that
	\begin{equation*}
		\left(1 + \frac{4}{R} \right)^{\log t} 
		 \le \left(1 + \frac{\ln(1 + \epsilon)}{\log t} \right)^{\log t}
		 \le e^{\ln(1 + \epsilon)}
		 = 1 + \epsilon \,.
	\end{equation*}
	
	For ${\algidx} = 0$ we have $\awm_{\algidx} = \wm_{\algidx}$
	and the inequality holds trivially. Assume the inequality holds for ${\algidx}$. We will
	show that it also holds for ${\algidx} + 1$. 
	
	First we prove the lower bound on $\awm_{{\algidx}+1}[i,j]$.
	Let $C_{{\algidx}+1}$ be the exact min-plus product of $\awm_{\algidx}$ with itself, i.e., 
	$C_{{\algidx}+1} = \awm_{\algidx} \mult \awm_{\algidx}$. Let $k_c$ be the minimizing index 
	such that $C_{{\algidx}+1}[i,j] = \min_{1 \leq k \leq n} (\awm_{\algidx}[i, k] 
	+ \awm_{\algidx}[k, j]) = \awm_{\algidx}[i, k_c] + \awm_{\algidx}[k_c, j]$.
	By the definition of the min-plus product
	\begin{equation}\label{eq:mp1}
		\wm_{{\algidx}+1}[i,j] = \min_{1 \leq k \leq n} (\wm_{\algidx}[i, k] + \wm_{\algidx}[k, j]) 
		\le \wm_{\algidx}[i, k_c] + \wm_{\algidx}[k_c, j]\,.
	\end{equation}
	By the induction hypothesis and the definition of $k_c$ we have
	\begin{equation}\label{eq:induction1}
		\wm_{\algidx}[i, k_c] + \wm_{\algidx}[k_c, j]
		\le \awm_{\algidx}[i, k_c] + \awm_{\algidx}[k_c, j]
		= C_{{\algidx}+1}[i,j]\,.
	\end{equation}
	By Theorem~\ref{th:amp} the values of $\awm_{{\algidx}+1}$ can only be larger than
	the values in $C_{{\algidx}+1}$, i.e.,
	\begin{equation}\label{eq:amp1}
		C_{{\algidx}+1}[i,j] \le \awm_{{\algidx}+1}[i,j]\,.
	\end{equation}
	Combining Equations~\eqref{eq:mp1},~\eqref{eq:induction1}, and~\eqref{eq:amp1}
	yields the claimed lower bound,
	\begin{equation*}
		\wm_{{\algidx}+1}[i,j] \le \awm_{{\algidx}+1}[i,j]\,.
	\end{equation*}
	
	Next we prove the upper bound on $\awm_{{\algidx}+1}[i,j]$.
	Let $k_d$ be the minimizing index such that $\wm_{{\algidx}+1}[i,j] = \wm_{\algidx}[i, k_d] + \wm_{\algidx}[k_d, j]$.
	Theorem~\ref{th:amp} gives the error from one call of $\amp$,
	i.e., the error in the entries of $\awm_{{\algidx}+1}$ compared to the entries of 
	$C_{{\algidx}+1}$. We have
	\begin{equation}\label{eq:amp2}
		\awm_{{\algidx}+1}[i,j] \le \left(1 + \frac{4}{R} \right) C_{{\algidx}+1}[i,j]\,.
	\end{equation}
	By the definition of the min-plus product we know that
	\begin{equation}\label{eq:mp2}
		C_{{\algidx}+1}[i,j] \le \awm_{\algidx}[i,k_d] + \awm_{\algidx}[k_d, j]\,.
	\end{equation}
	By the induction hypothesis and the definition of $k_d$ we can reformulate
	the error obtained in the first ${\algidx}$ iterations of Algorithm~\ref{alg} as follows:
	\begin{align}\label{eq:induction2}
		\awm_{\algidx}[i,k_d] + \awm_{\algidx}[k_d, j] &\le \left(1 + \frac{4}{R} \right)^{\algidx} \wm_{\algidx}[i, k_d] +
		\left(1 + \frac{4}{R} \right)^{\algidx} \wm_{\algidx}[k_d, j] \,, \notag \\
		&= \left(1 + \frac{4}{R} \right)^{\algidx} \left( \wm_{\algidx}[i, k_d] + \wm_{\algidx}[k_d, j] \right) \,, \notag \\
		&= \left(1 + \frac{4}{R} \right)^{\algidx} \wm_{{\algidx}+1}[i,j]\,.
	\end{align}
	Combining Equations~\eqref{eq:amp2},~\eqref{eq:mp2}, and~\eqref{eq:induction2}
	yields the upper bound
	\begin{equation*}
		\awm_{{\algidx}+1}[i,j] \le \left(1 + \frac{4}{R} \right)^{{\algidx}+1} \wm_{{\algidx}+1}[i,j]\,.\qedhere
	\end{equation*}
\end{proof}

Once we have computed an approximation of the matrix $\wm^t$, we extract from it the minimal entry of each row to obtain an approximation of $ \vt_t (x) $.
Here we use the equivalence between the minimum entry of row $x$ of $\wm^t$ and $\vt_t(x)$ established in Proposition~\ref{prop:rowmin}.
Remember that $\vt_t(x) / t$ approaches $\val(x)$ for $t$ large enough and later on we want to use the approximation of $\vt_t(x)$ to obtain an approximation
of the minimum cycle mean $\val(x)$.

\begin{lemma}\label{lem:approxvt}
The value $\at_t(x) := \min_{y \in V} \awm[x, y]$ approximates $\vt_t(x)$ with
$\vt_t(x) \le \at_t(x) \le (1 + \epsilon) \vt_t(x)\,$.
\end{lemma}
\begin{proof}
Let $y_{f}$ and $y_{d}$ be the indices where the $x$-th rows of
$\awm$ and $\wm^t$ obtain their minimal values, respectively, i.e.,
\begin{equation*}
	y_{f} := \argmin_{y \in V} \awm[x, y]
	\quad \text{and} \quad 
	y_{d} := \argmin_{y \in V} \wm^t[x, y]\,.
\end{equation*}
By these definitions and Lemma~\ref{lem:alg} we have
\begin{equation*}\label{eq:def1}
	\vt_t(x) = \wm^t[x, y_{d}] \le \wm^t[x, y_{f}] 
	\le \awm[x, y_{f}] = \at_t(x)\,
\end{equation*}
and
\begin{equation*}\label{eq:def2}
	\at_t(x) = \awm[x, y_{f}] \le \awm[x, y_{d}] 
	\le (1 + \epsilon)\wm^t[x, y_{d}] \,.\qedhere
\end{equation*}
\end{proof}

\subsection{Approximating the minimum cycle mean}

We now add the next building block to our algorithm.
So far, we can obtain an approximation $ \at_t (x) $ of $\vt_t(x)$ for any $ t $ that is a power of two.
We now show that $ \vt_t(x) / t $ is itself an approximation of the minimum cycle mean $ \val (x) $ for $ t $ large enough.
Then we argue that $ \at_t (x) / t $ approximates the minimum cycle mean $ \val (x) $ for $ t $ large enough.
This value of $ t $ bounds the number of iterations of our algorithm.
A similar technique was also used in \cite{ZwickP1996} to bound the number of iterations of the value iteration algorithm for the two-player mean-payoff game.

We start by showing that $\vt_t(x) / t $ differs from $ \val(x) $ by at most $n W / t$ \emph{for any $t$}.
Then we will turn this additive error into a multiplicative error by choosing a large enough value of $ t $.
A multiplicative error implies that we have to compute the solution exactly for $\val(x) = 0$.
We will use a separate procedure to identify all vertices $x$ with $\val(x) = 0$ and compute
the approximation only for the remaining vertices. Note that $\mu (x) > 0$
implies $\mu (x) \ge 1/n$ because all edge weights are integers.

\begin{lemma}\label{lem:nWbound}
	For every $x \in V$ and every integer $t \ge 1$ it holds that
	\begin{equation*}
		t \cdot \val(x) - n W \le \vt_t(x) \le t \cdot \val(x) + n W\,.
	\end{equation*}
\end{lemma}

\begin{proof}
We first show the lower bound on $\vt_t(x)$.
Let $P$ be a path of length $t$ starting at $x$ with weight $\vt_t(x)$. 
Consider the cycles in $P$ and let $E'$ be the multiset of the edges in $P$ 
that are in a cycle of $P$. There can be at most $n$ edges that are not
in a cycle of $P$, thus there are at least $\max(t-n, 0)$ edges in $E'$.
Since $\val(x)$ is the minimum mean weight of any cycle reachable from $x$, 
the sum of the weight of the edges in $E'$ can be bounded below by $\val(x)$ times the
number of edges in $E'$. Furthermore, the value of $\val(x)$ can be at most $W$. 
As we only allow nonnegative edge weights, the sum of the weights of the edges 
in $E'$ is a lower bound on $\vt_t(x)$. Thus we have
\begin{equation*}
	\vt_t(x) \ge \sum_{e \in E'} w(e) 
	\ge (t-n) \val(x) \ge t \cdot \val (x) - n \cdot \val (x) \ge t \cdot \val(x) - n W\,.
\end{equation*}

Next we prove the upper bound on $\vt_t(x)$.
Let $l$ be the length of the shortest path from $x$ to a vertex $y$ in a minimum
mean-weight cycle $C$ reachable from $x$ (such that only $y$ is both in the
shortest path and in $C$). Let $c$ be the length of $C$. 
Let the path $Q$ be a path of length $t$ that consists
of the shortest path from $x$ to $y$, $\lfloor (t - l) / c \rfloor$ rounds on $C$,
and $t - l - c\lfloor (t - l) / c \rfloor$ additional edges in $C$.
By the definition of $\vt_t(x)$, we have $\vt_t(x) \le w(Q)$.
The sum of the length of the shortest path from $x$ to $y$ and the number of
the remaining edges of $Q$ not in a complete round on $C$ can be at most $n$ because in a graph
with nonnegative weights no shortest path has a cycle and no vertices in $C$
except $y$ are contained in the shortest path from $x$ to $y$. Each of these edges 
has a weight of at most $W$. The mean weight of $C$ is $\val(x)$, thus
the sum of the weight of the edges in all complete rounds on $C$ is 
$\val(x) \cdot c\lfloor (t - l) / c \rfloor \le \val(x) \cdot t$. Hence we have
\begin{equation*}
	\vt_t(x) \le w(Q) \le t \cdot \val(x) + n W\,. \qedhere
\end{equation*}
\end{proof}

In the next step we show that we can use the fact that $\vt_t(x) / t$ is an
approximation of $\val(x)$ to obtain a 
$(1 + \epsilon)$-approximation $\appr(x)$ of $\val(x)$ even if 
we only have an approximation $\at_t(x)$ of $\vt_t(x)$ with $(1 + \epsilon)$-error. 
We exclude the case $\val(x) = 0$ for the moment.

\begin{lemma}\label{lem:approxv}
	Assume we have an approximation $\at_t(x)$ of $\vt_t(x)$ such that
	$\vt_t(x) \le \at_t(x) \le (1 + \epsilon) \vt_t(x)$ for $0 < \epsilon \le 1/2$. 
	If 
	\begin{equation*}
		t \ge \frac{n^2 W}{\epsilon}\,,
		\quad
		\val(x) \ge \frac{1}{n}\,,
		\quad \text{and} \quad 
		\appr(x) := \frac{\at_t(x)}{(1-\epsilon)t}\,,
	\end{equation*}
	then 
	\begin{equation*}
		\val(x) \le \appr(x) 
		\le (1 + 7 \epsilon) \val(x)\,.
	\end{equation*}
\end{lemma}

\begin{proof}
	We first show that $\appr(x)$ is at least as large as $\val(x)$.
	From Lemma~\ref{lem:nWbound} we have $\vt_t(x) \ge 
	t \cdot \val(x) - n W$. As $t$ is chosen large enough,
	\begin{equation*}
		\frac{\vt_t(x)}{t} \ge \val(x) - \frac{n W}{t} 
		\ge \val(x) - \frac{\epsilon}{n}
		\ge \val (x) - \epsilon \val (x)
		\ge (1 - \epsilon) \val(x)\,.
	\end{equation*}
	Thus, by the assumption $\vt_t(x) \le \at_t(x)$ we have
	\begin{equation*}
		\val(x) \le \frac{\at_t(x)}{(1-\epsilon)t} = \appr(x)\,.
	\end{equation*}

	For the upper bound on $\appr(x)$ we use the inequality 
	$\vt_t(x) \le t \cdot \val(x) + n W$ from Lemma~\ref{lem:nWbound}. 
	As $t$ is chosen large enough,
	\begin{equation*}
		\frac{\vt_t(x)}{t} \le \val(x) + \frac{n W}{t} 
		\le \val(x) + \frac{\epsilon}{n}
		\le (1 + \epsilon) \val(x)\,.
	\end{equation*}
	With $\at_t(x) \le (1 + \epsilon) \vt_t(x)$ this gives
	\begin{equation*}
		\appr(x) = \frac{\at_t(x)}{(1-\epsilon)t} 
		\le \frac{(1 + \epsilon)^2}{(1-\epsilon)} \val(x)\,.
	\end{equation*}
	It can be verified by simple arithmetic that for $\epsilon > 0$ 
	the inequality $\epsilon \le 1/2$ is equivalent to
	\begin{equation*}
		\frac{(1 + \epsilon)^2}{(1-\epsilon)} 
		\le (1 + 7 \epsilon)\,.\qedhere
	\end{equation*}
\end{proof}

As a last ingredient to our approximation algorithm, we design a procedure that deals with the special case that the minimum cycle mean is $ 0 $.
Since our goal is an algorithm with multiplicative error, we have to be able to compute the solution exactly in that case.
This can be done in linear time because the edge-weights are nonnegative.

\begin{proposition}\label{prop:val0}
Given a graph with nonnegative integer edge weights, we can find out all vertices $ x $ such that $ \val (x) = 0 $ in time $ O (m) $.
\end{proposition}

\begin{proof}
Note that in the case of nonnegative edge weights we have $ \val (x) \geq 0 $.
Furthermore, a cycle can only have mean weight $ 0 $ if all edges on this cycle have weight $ 0 $.
Thus, it will be sufficient to detect cycles in the graph that only contain edges that have weight $ 0 $.

We proceed as follows.
First, we compute the strongly connected components of $ G $, the original graph.
Each strongly connected component $ G_i $ (where $ 1 \leq i \leq k $) is a subgraph of $ G $ with a set of vertices $ V_i $ and a set of edges $ E_i $.
For every $ 1 \leq i \leq k $ we let $ G_i^0 = (E_i^0, V_i) $ denote the subgraph of $ G_i $ that only contains edges of weight $ 0 $, i.e., $ E_i^0 = \{ e \in E_i | w (e) = 0 \} $.
As argued above, $ G_i $ contains a zero-mean cycle if and only if $ G_i^0 $ contains a cycle.
We can check whether $ G_i^0 $ contains a cycle by computing the strongly connected components of $ G_i^0 $:
$ G_i^0 $ contains a cycle if and only if it has a strongly connected component of size at least $ 2 $
(we can assume w.l.o.g. that there are no self-loops).
Let $ Z $ be the set of all vertices in strongly connected components of $ G $ that contain a zero-mean cycle.
The vertices in $Z$ are not the only vertices that can reach a zero-mean cycle.
We can identify all vertices that can reach a zero-mean cycle by performing a 
linear-time graph traversal to identify all vertices that can reach $Z$. 

Since all steps take linear time, the total running time of this algorithm is $ O (m) $.
\end{proof}

Finally, we wrap up all arguments to obtain our algorithm for approximating the minimum cycle mean.
This algorithms performs $ \log t $ approximate min-plus matrix multiplications to compute an approximation of $ \wm^t $ and $ \vt_t (x) $.
Lemma~\ref{lem:approxv} tells us that $ t = n^2 W / \epsilon $ is just the right number
to guarantee that our approximation of $ \vt_t (x) $ can be used to obtain an approximation of $ \val (x) $.
The value of $ t $ is relatively large but the running time of our algorithm depends on $ t $ only in a logarithmic way.

\begin{theorem}\label{th:nonnegative}
Given a graph with nonnegative integer edge weights, we can compute
an approximation $ \appr (x) $ of the minimum cycle mean for every 
vertex $ x $ such that $ \val (x) \leq \appr (x) \leq (1 + \epsilon) \val (x)$
for $0 < \epsilon \le 1$
in time 
	\begin{equation*}
		O\left(\frac{ n^{\omega} }{\epsilon}
		\log^3 \left( \frac{n W}{\epsilon} \right) 
		\log^2 \left( \frac{\log \left(\frac{n W}{\epsilon}\right)}{\epsilon} \right)
		\log (n) 
		\right)
		= \widetilde O \left( \frac{ n^{\omega} }{\epsilon}
		\log^3 \left( \frac{n W}{\epsilon} \right) 
		\right)\,.
	\end{equation*}
\end{theorem}

\begin{proof}
First we find all vertices $x$ with $\val(x) = 0$. By Proposition~\ref{prop:val0}
this takes time $O(n^2)$ for $m = O(n^2)$. For the remaining vertices $x$ we
approximate $\val(x)$ as follows.

Let $\epsilon':= \epsilon / 7$. If we execute Algorithm~\ref{alg} with
weight matrix $\wm$, error bound $\epsilon'$ and $t$ such that $t$ is the 
smallest power of two with $t \ge n^2 W / \epsilon'$, 
we obtain a $(1+\epsilon')$-approximation $\awm[x,y]$ of $\wm^t[x, y]$ 
for all vertices $x$ and $y$ (Lemma~\ref{lem:alg}).
By calculating for every $x$ the minimum entry of $\awm[x,y]$ over all $y$ we
have a $(1 + \epsilon')$-approximation of $\vt_t(x)$ (Lemma~\ref{lem:approxvt}).
By Lemma~\ref{lem:approxv} $\appr(x) := \at_t(x) / ((1-\epsilon')t)$ 
is for this choice of $t$ an approximation of $\val(x)$ such that 
$\val(x) \le \appr(x) \le (1 + 7 \epsilon') \val(x)$. 
By substituting $\epsilon'$ with $\epsilon / 7$ we get $\val(x) \le \appr(x) 
\le (1 + \epsilon) \val(x)\,$ i.e., a $(1 + \epsilon)$-approximation of $\val(x)$.

By Lemma~\ref{lem:alg} the running time of Algorithm~\ref{alg} for 
$ t = 2^{\lceil \log (n^2 W / \epsilon') \rceil} = O (n^2 W / \epsilon) $
is
\begin{align*}
	O\left(\frac{ n^{\omega} }{\epsilon}
	\log^2 \left( \frac{n^2 W}{\epsilon} \right) 
	\log \left( \frac{n^2 W^2}{\epsilon} \right)
	\log^2 \left( \frac{\log \left(\frac{n^2 W}{\epsilon}\right)}{\epsilon} \right)
	\log (n) 
	\right)\,.
\end{align*}
With $\log(n^2 W) \le \log((n W)^2) = O(\log(n W))$ we get that Algorithm~\ref{alg} runs in time
\begin{equation}
	O\left(\frac{ n^{\omega} }{\epsilon}
	\log^3 \left( \frac{n W}{\epsilon} \right) 
	\log^2 \left( \frac{\log \left(\frac{n W}{\epsilon}\right)}{\epsilon} \right)
	\log (n)
	\right)\,.\qedhere
\end{equation}
\end{proof}

\section{Open problems}

We hope that this work draws attention to the problem of approximating
the minimum cycle mean. It would be interesting to study whether
there is a faster approximation algorithm for the minimum cycle mean problem,
maybe at the cost of a worse approximation.
The running time of our algorithm immediately improves if faster algorithms
for classic matrix multiplication, min-plus matrix multiplication or 
approximate min-plus multiplication are found. However, a different
approach might lead to better results and might shed new light on how well the
problem can be approximated. Therefore it would be interesting to remove the dependence
on fast matrix multiplication and develop a so-called combinatorial
algorithm.

Another obvious extension is to allow negative edge weights in the input graph.
Furthermore, we only consider the minimum cycle mean problem, while it might
be interesting to actually output a cycle with approximately optimal mean weight.

\bibliographystyle{eptcs}
\bibliography{literature}
\end{document}